\documentclass[letterpaper,10pt,conference]{ieeeconf}

\usepackage{cite}
\usepackage{comment}
\usepackage{makecell}
\usepackage{xcolor}
\usepackage[normalem]{ulem}
\makeatletter\let\NAT@parse\@undefined\makeatother
\usepackage[colorlinks=true, linkcolor=blue, urlcolor=blue, citecolor=black]{hyperref}
\usepackage{cntsys}

\IEEEoverridecommandlockouts
\title{\LARGE \bf Exact and Approximate Convex Reformulation of Linear Stochastic Optimal Control with Chance Constraints}

\author{Tanmay Dokania$^{1}$ and Yashwanth Kumar Nakka$^{1}$ 
\thanks{$^{1}$Authors are with the Daniel Guggenheim School of Aerospace Engineering at the Georgia Institute of Technology, Atlanta, GA, 30332 USA. Email: [tdokania3,ynakka3]@gatech.edu
}}
\begin{document}
\maketitle
\thispagestyle{empty}
\pagestyle{empty}
\begin{abstract}
In this paper, we present an equivalent convex optimization formulation for discrete-time stochastic linear systems subject to linear chance constraints, alongside a tight convex relaxation for quadratic chance constraints. By lifting the state vector to encode moment information explicitly, the formulation captures linear chance constraints on states and controls across multiple time steps exactly, without conservatism, yielding strict improvements in both feasibility and optimality. For quadratic chance constraints, we derive convex approximations that are provably less conservative than existing methods. We validate the framework on minimum-snap trajectory generation for a quadrotor, demonstrating that the proposed approach remains feasible at noise levels an order of magnitude beyond the operating range of prior formulations.
\end{abstract}
\section{Introduction}

Stochastic Motion Planning refers to the problem of constructing a path from an initial to a final state, subject to the uncertainty from sensors, dynamics, or obstacles, while avoiding them and minimizing a cost function~\cite{nakka_trajectory_2019,nakka_trajectory_2023,lavalle_planning_2006}. 
Uncertainty leads to a distribution over states, and therefore a notion of `tube' becomes necessary, which is not accounted for by deterministic planners, leading to incorrect safety and quality assessment~\cite{yu_gaussian_2023}. Examples of such systems include spacecrafts encountering interstellar objects~\cite{tsukamoto_neural-rendezvous_2024} and Unmanned Aerial Vehicles (UAVs) under turbulent winds~\cite{singh_robust_2017}.

A motion planning problem under optimality and chance constraints can be formulated as a stochastic optimal control problem. Optimization-based control of linear stochastic systems 
serves as a critical backbone for several nonlinear approaches that iteratively linearize around different points~\cite{lew_chance-constrained_2020,nakka_trajectory_2023}, or run multiple in parallel~\cite{mustafa_probabilistic_2023}. In this simplified setting, the problem is convex and can be solved in real time; therefore, a better understanding of this setup can significantly improve algorithms upstream.

Covariance control refers to finding an optimal control law that drives a distribution to another distribution with specified covariance.
Such approaches with a feedback structure can be used to solve the stochastic motion planning problem.
Without chance constraints, the problem simplifies because the mean and covariance optimization can be decoupled; as a result, theoretical results show that feedback laws are optimal for steering distributions~\cite{liu_optimal_2025}. Additionally, a lossless convex relaxation is used to solve the problem under both multiplicative and additive noise efficiently~\cite{balci_covariance_2023}.
Decoupling the mean and covariance, however, also decouples the risk-performance tradeoff, as tighter feedback laws navigate a narrow corridor rather than an open space. This can be formulated as a set of chance constraints. In the study of linear chance constraints, feedback laws are assumed optimal and used to formulate an optimization problem that yields nonconvex constraints. They are solved via linearization~\cite{rapakoulias_discrete-time_2023} or successive convexification~\cite{pilipovsky_computationally_2024}. To the best of the author's knowledge, there are no theoretical results on the optimality of feedback laws under a chance-constrained setting for the covariance control problem.

We propose a lifted representation method, inspired by~\cite{nakka_trajectory_2023}, that captures the distribution of state and control input across different time steps. First, we show that this allows for the exact reformulation of linear chance constraints as second-order cone constraints under a Gaussian distribution, the dynamic constraints remain affine, and the representation of general linear functions involving states and controls from different timesteps remains possible. 
Secondly, we develop convex approximations of quadratic chance constraints as linear-matrix and quadratic constraints and compare their approximation tradeoff.
Finally, we show that our approach outperforms other approaches that can represent state and control chance constraints~\cite{rapakoulias_discrete-time_2023,lew_chance-constrained_2020} in terms of feasibility and optimality, even with an order-of-magnitude larger noise standard deviation.

\textit{\textbf{Notation:}} Random variables are denoted by boldface letters (e.g., $\state k, \cntrl k, \basv, \noise k$). The cones of $n \times n$ positive definite and positive semi-definite matrices are denoted by $\pd{n}$ and $\psd{n}$, respectively, and $I_n$ represents the $n \times n$ identity matrix. Norm of a vector $x$ with respect to a matrix $A\in\pd{n}$ is denoted as $\quadform{A}{x} = x^\top A x$. The operators $\Tr(\cdot)$, $\vc{\cdot}$, and $\otimes$ denote the trace, the vectorization of a matrix, and the Kronecker product, respectively. For statistical operations, $\Ex{\cdot}$ and $\Vr{\cdot}$ denote the expectation and variance. A Gaussian random vector with mean $\mu$ and covariance matrix $\Sigma$ is denoted by $\N{\mu}{\Sigma}$. Finally, $\Phi_{\mathcal N}(\cdot)$ and $\Phi_{\chi^2_n}(\cdot)$ represent the cumulative distribution functions (CDFs) of the standard Gaussian distribution and the chi-squared distribution with $n$ degrees of freedom, respectively.

\begin{table*}[t] 
    \centering
    \caption{Overview and Comparison of different Convex Optimization-based solvers for Linear Stochastic Optimal Control Problem. The decision variables for all the approaches include the state and control mean. For the stochastic component, they assume a feedback structure $K$ or optimize over the state-control covariance $\Vr{[\state{k}, \cntrl{k}] }:= \Sigma^{xu}$. 
    $\checkmark$ and $\approx$ demarcate exact and approximate representation capability, while $\times$ denotes inability to represent the said constraint or property.
    }
    \label{tab:lcc_comparison}. 
    \begin{tabular}{c l l c c c c c c c c}
        \toprule
        \textbf{Method} & 
        \makecell{\textbf{Decision} \\ \textbf{Variables}} & 
        \makecell{\textbf{LCC} \\ \textbf{Reform.}} & 
        \makecell{\textbf{Problem} \\ \textbf{Type}} & 
        \makecell{\textbf{Distribution} \\ \textbf{Optim.}} & 
        \makecell{\textbf{Time} \\ \textbf{Varying}} & 
        \makecell{\textbf{Term.} \\ \textbf{Cov.}} & 
        \makecell{\textbf{Quad.} \\ \textbf{CC}} & 
        \makecell{\textbf{Control} \\ \textbf{LCC}} & 
        \makecell{\textbf{State} \\ \textbf{LCC}} & 
        \makecell{\textbf{Multi-step} \\ \textbf{LCC}} \\
        \midrule
         \cite{okamoto_optimal_2018}& $\mux{}, \muu{}, K$ & Exact SOCP & Control & $\checkmark$ & $\checkmark$ & $\checkmark$ & $\times$ & $\times$ & $\checkmark$ & $\checkmark$ \\
        \cite{rapakoulias_discrete-time_2023}  & $\mux{}, \muu{}, \Sigma^{xu}$ & Approx. Linear & Control & $\checkmark$ & $\checkmark$ & $\checkmark$ & $\times$ & $\approx$ & $\approx$ & $\times$ \\
        \cite{lew_chance-constrained_2020}   & $\mux{}, \muu{}$ & Exact Linear & Control & $\times$ & $\checkmark$ & $\times$ & $\times$ & $\checkmark$ & $\checkmark$ & $\times$ \\
        \cite{vitus_feedback_2011,skaf_design_2010}  & $\mux{}, \muu{}, K$ & Exact SOCP & Control & $\checkmark$ & $\times$ & $\times$ & $\times$ & $\times$ & $\checkmark$ & $\times$ \\
        \textbf{Ours}    & $\mux{}, \muu{}, \varx{}, \varu{}$ & Exact SOCP & \textbf{Motion Plan.} & $\checkmark$ & $\checkmark$ & $\checkmark$ & \textbf{$\approx$} & $\checkmark$ & $\checkmark$ & $\checkmark$ \\
        \bottomrule
    \end{tabular}
\end{table*}
\section{Problem Formulation}
This section presents the stochastic optimal control problem with chance constraints. The following stochastic linear time-varying system is considered for $k \in \itt {N-1}$:
\begin{equation}
\begin{split}
\state {k+1} &= A_k \state k + B_k \cntrl k + D_k\noise k, \label{eq:sys_dyn}\\
\state 0 &\sim \N{\mu_0^\star}{\Sigma_0},
\qquad
\noise k \sim \N{0}{I_{n_w}},
\end{split}
\end{equation}
here, $N$ is the horizon, $\{\state k\}_{k=0}^N, \{\cntrl k\}_{k=0}^{N-1}, \{\noise k\}_{k=0}^{N-1}$ are the state, control and noise random processes over $\R^{n_x}, \R^{n_u}, \R^{n_w}$ respectively,
and $A_k\in\R^{n_x \times n_x}$, $B_k \in \R^{n_x \times n_u}$, $D_k \in \R^{n_x \times n_w}$ are the system, control and disturbance matrices respectively.
\begin{asmp}
    The initial state $(\state 0)$ and the process noise $(\noise k)$ at each time step are independent Gaussian R. Vs., implying $\Ex{\noise i \noise j^\top} = \delta_{ij} I_{n_w},\ \forall i,j \in \itt{N-1}$ and $\Ex{\noise i \state 0^\top } = 0,\ \forall i \in \itt{N-1}$.
    The control input $(\cntrl k)$ is an R.V. that can be represented as a linear combination of the independent R.Vs. realized up to time step $k$, i.e., $\state 0, \noise 0, \dots, \noise {k-1}$. This class of causal and affine policies is denoted by $\Pi_{\text{affine}}$.
\label{asmp:independence_represent}
\end{asmp}
As the state and control are random processes, chance constraints are required to give probabilistic guarantees. The constraints generally considered are linear chance constraints on state or control at a given time step:
\begin{equation*}
    \Pr {a^\top \state i \le b} \ge 1 - \epsilon,\quad \Pr {\alpha^\top \cntrl i \le b} \ge 1 - \epsilon,
\end{equation*}
here, $a \in \R^{n_x}, \alpha \in \R^{n_u}, b \in \R$, and $\epsilon>0$ measures the confidence of the constraint. We propose the consideration of mixed-multistep linear chance constraints that allow for representing chance constraints across multiple time steps and mixed with control constraints of the following type,
\begin{equation}
    \Pr{a_i^\top \state{0:N} + \alpha_i^\top\cntrl{0:N-1} \le b_i} \ge 1-\epsilon_i,  \label{eq:def_lcc_mixed_multi}
\end{equation}
here, $a_i\in\R^{(N+1) n_x}, \alpha_i \in \R^{N n_u}, b_i\in \R, \epsilon_i  \in [0, 0.5)$, $\state{0:N} = [\state 0^\top, \state 1 ^\top, \dots, \state N ^\top]^\top$, $\cntrl{0:N} = [\cntrl 0^\top, \cntrl 1 ^\top, \dots, \cntrl {N-1} ^\top]^\top$  $\forall \ i \in \itt{m}$.
\begin{rmrk}
    Polytropic chance constraints can be described using multiple linear chance constraints, where the risk associated is allocated using Boole's and a bilevel optimization problem~\cite{vitus_feedback_2011}.
\end{rmrk}
Representation of state-control-coupled constraints allows us to define a linearized version of power constraints on a motor or friction-circle constraints for tire grip.

We consider centered quadratic chance constraints, which are represented as
\begin{equation}
\Pr{\quadform{ \termQ^{-1}}{ (\state i-\Ex{\state i})} \le 1}\ge 1-\epsilon, \label{eq:termQ}
\end{equation}
here, $\termQ \in\pd{n_x}, \epsilon \in [0,1]$ defines the ellipsoid that the state must be in, with a high confidence measure  

In covariance control problems, a constraint on the terminal covariance is imposed with $\Sigma_f \in \pd{n_x}$,
\begin{equation}
    \Vr{x_N} \preceq \Sigma_f \label{eq:cov_control}.
\end{equation}

For motion planning problems, waypoints that the robot needs to follow through at specified timesteps $T_{wp}\subseteq \itt{N} $, can be described as soft constraints through the cost function $\quadform{Q_\star}{(\Ex{\state k} - \mux k ^\star)}$
where $ \mu_k^\star \in \R^{n_x}$.

Under quadratic cost on the system state and control input, the finite horizon optimal control problem with chance constraints is formulated as:
\begin{problem}[\textbf{Linear Stochastic Optimal Control}]\label{prob:LSOC}
\begin{equation}
\scalebox{1}{$
\begin{aligned}
    \min_{\state {}, \cntrl{}}  J(\state{},\cntrl{}) =
    \sum_{i \in T_{wp}} \quadform{ Q_\star} {(\Ex{\state i} - \mux i ^\star)}
    \\
+\Ex{
\sum_{k=0}^{N-1}
 \quadform{Q_k}{\state k} + \quadform{R_k}{\cntrl k})+\quadform{Q_N}{(\state N-x^\star)} } 
\label{eq:prob1:orig_cost}
\end{aligned}$}
\end{equation}
subject to $\forall\ k \in \{0, 1, \dots, N-1\}$ 
\begin{align}
{\eqref{eq:sys_dyn}, \eqref{eq:def_lcc_mixed_multi}, \eqref{eq:termQ},\eqref{eq:cov_control}}\notag
\end{align} 
here, $Q_k \in \psd{n_x}, R_k \in \psd{n_u}$ for $k \in \itt{N-1}$ and $Q_N, Q^\star \in \psd{n_x}, x^{\star}\in\R^{n_x}$.
\end{problem}
This optimization problem is hard to numerically solve as it is formulated over R.Vs. Hence, we propose a novel reformulation, inspired by~\cite{nakka_trajectory_2023}, under which the problem is convex and the constraints are reformulated: linear chance constraints exactly as second-order cone constraints, quadratic chance constraints approximately as linear-matrix or quadratic constraints, and the covariance constraint exactly as a linear-matrix constraint.
\section{Convex Reformulation}
In this section, we first present the lifted state and control representation, followed by the dynamics, cost, and constraints. This section concludes with the key result that the proposed formulation is an exact reformulation of~\Cref{prob:LSOC}.
\subsection{Stochastic basis}
Under the dynamics \eqref{eq:sys_dyn} and~\Cref{asmp:independence_represent}, it follows that the state and control remain jointly Gaussian at every time step. Therefore, we propose using scalar Gaussians to capture the correct marginal and conditional distributions, rather than propagating covariances.

As $\state 0 \sim \N{\mu_0^\star}{\Sigma_0}$, it can be represented as $\state 0 = \mu_0^\star + \mvarx 0^\star \xiv 0$, where $ \xiv 0 \sim \N{0}{I_{n_x}}$ and $\mvarx 0^\star \mvarx 0^{\star\top} = \Sigma_0$, thereby, requiring $n_x$ scalar Gaussians. Additionally, at each time step $k$, $\noise k$ introduces additional $n_w$ independent scalar Gaussians. Therefore, we define the stochastic basis recursively as follows:
\begin{equation}
\xiv {k+1} := \begin{bmatrix}
    \xiv k \\ \noise k
\end{bmatrix},\ \ \forall\ k \in \{0, 1, \ldots, N-1 \} \label{eq:stoch_basis_def}.
\end{equation}
Hence, $\xiv k \sim \N{0}{I_{\bas k}}$, where $\bas k := n_x + k n_w$.
\subsection{State representation}
The R.V. $\state k$ is represented using two vectors,
$
    \mux{k} \in \R^{n_x}, \varx k \in \R^{n_x\bas k} .
$
The vector $\varx k$ is the vectorized form of the matrix $\mvarx k \in \R^{n_x \times \bas k}$. Hence, there is a bijective mapping between them:
\begin{equation}
\begin{split}
    \varx k &= \vc{ {\mvarx k}}, \\ \mvarx k &= (\vc{I_{\bas k}}^\top \otimes I_{n_x}) (I_{\bas k} \otimes \varx k)  .
\end{split}
\end{equation}
The state at step $k$ is represented as:
\begin{equation}
    \state k = \mux k + {\mvarx k} \xiv k \label{eq:state_rep}.
\end{equation}
It follows that:
\begin{equation}
    \Ex{\state k} = \mux k, \ \Vr{\state k} = \mvarx{k} \mvarx{k}^\top \label{eq:state_mean_cov}.
\end{equation}
\subsection{Control representation}
Following~\Cref{asmp:independence_represent}, the control input at step $k$ depends on the basis elements realized up to the time step, i.e., $\xiv k$. Therefore, it retains the causal property.
Similarly, $\cntrl k$ is represented using two vectors
$
    \muu{k} \in \R^{n_u}, \varu k \in \R^{n_u\bas k} 
$.
The matrix representation ${\mvaru k} \in \R^{n_u \times \bas k}$ is uniquely related to the vector representation $\varu k$ as follows:
\begin{equation}
\begin{split}
   \varu k &=  \vc{ {\mvaru k}},\\
    \mvaru k &= (\vc{I_{\bas k}}^\top \otimes I_{n_u}) (I_{\bas k} \otimes \varu k)  .
\end{split}
\end{equation}
The control at step $k$ is represented as:
\begin{equation}
    \cntrl k = \muu k + {\mvaru k} \xiv k \label{eq:control_rep}.
\end{equation}
\begin{lemma}\label{l1:linear_chance}
    Let $\mathbf{y}$ be a scalar R.V. If it is representable in the basis $\xiv k$, i.e., $\mathbf y = y_0 + y \xiv k$, where $y_0 \in \R, y \in \R^{1 \times \bas k}$ then, 
    a chance constraint $\Pr{\mathbf y \le b}\ge 1-\epsilon$ can be exactly represented as
    $
        y_0 + \Phi_{\mathcal N}^{-1} (1-\epsilon) \|y \|_I \le b ,
    $ 
    which is a second-order cone exact reformulation of the chance constraint for $0\le \epsilon < 0.5$.
\end{lemma}
\begin{proof}
   As $\mathbf {y} $ is represented as a linear combination of jointly Gaussian R.V. $\xiv k$, it is a Gaussian R.V. Following the linearity of expectation and the fact that $\xiv k \sim \N{0}{I_{\bas k}}$, we obtain $\Ex{\mathbf y} = y_0, \ \Vr{\mathbf y} = \|y \|^2_2 $. Following the definition of $\Phi_{\mathcal N}$, we have
    \begin{equation}
        \Pr {y \le b } = \Phi_{\mathcal N} \left( \frac{b-y_0}{\|  y\|_I} \right).
    \end{equation}
    Using monotonicity of $\Phi_\mathcal{N}$, we get the above-mentioned constraint, which is a second-order cone for $0\le \epsilon < 0.5$.
\end{proof}
\begin{lemma}
    Under the representation \eqref{eq:state_rep}, \eqref{eq:control_rep} and ~\Cref{asmp:independence_represent}, the dynamics \eqref{eq:sys_dyn} is reformulated as:
    \begin{subequations}
\begin{align}
    \mux {k+1} &= A_k \mux k + B_k \muu k \label{eq:l2:mu_dynamics},\\
    \varx{k+1} &= \begin{bmatrix}
        (I_{\bas k} \otimes A_k) \varx k + (I_{\bas k} \otimes B_k) \varu k \\
        \vc{D_k}
    \end{bmatrix} \label{eq:l2:ran_dynamics},
\end{align}
subject to the initial condition $\mux 0 = \mux 0 ^\star, \ \varx 0 = \vc{\mvarx 0^\star}$. \label{eq:lifted_dynamics}
\end{subequations}
\end{lemma}
\begin{proof}

    Using linearity of expectation over \eqref{eq:sys_dyn}, the definition of $\mux k, \muu k$ in \eqref{eq:state_rep} and \eqref{eq:control_rep}, \eqref{eq:l2:mu_dynamics} follows. Under the initial condition specified for $\mux 0$ and $\varx 0$, \eqref{eq:sys_dyn} is expressed as
    \begin{equation}
    \begin{split}
     {\mvarx {k+1}} \xiv {k+1} 
      &= \begin{bmatrix}
          A_k {\mvarx k} + B_k {\mvaru k} & D_k
      \end{bmatrix} \begin{bmatrix}
          \xiv k \\ \noise k
      \end{bmatrix}.
    \end{split}
    \end{equation}
    Under the basis definition \eqref{eq:stoch_basis_def}, and the fact that individual elements of the basis are uncorrelated, we obtain $ {\mvarx {k+1}} = \begin{bmatrix}
          A_k {\mvarx k} + B_k {\mvaru k} & D_k
       \end{bmatrix}$. Application of the vectorization operation and using $\vc{A_k \mvarx k} = (I_{\bas k}\otimes A_k) {\varx k}$, results in \eqref{eq:l2:ran_dynamics}.

\end{proof}
\subsection{Linear Chance Constraint}
Given the mixed-multi-step linear chance constraint \eqref{eq:def_lcc_mixed_multi}, we introduce a scalar R.V. $
    \mathbf y^{(i)} := a_i^\top \state{0:N} + \alpha_i^\top \cntrl{0:N-1}
$.
The chance constraint is equivalent to $\Pr{\mathbf y^{(i)} \le b} \ge 1-\epsilon_i$.
Therefore, under the stochastic basis,
\begin{equation}
\begin{split}
    \mathbf y^{(i)} = &a_i^\top \mux{0:N} + \alpha_i^\top \muu{0:N-1} \\
    &+ a_i^\top \mvarx{0:N} \xiv N + \alpha_i^\top \mvaru{0:N-1} \xiv{N},
\end{split}
\end{equation}
here, 
$\mux{0:N}, \muu{0:N-1}, \mvarx{0:N}, \mvaru{0:N-1}$ are vertical stacks of $\mux i, \muu i, \mvarx i, \mvaru i$, respectively, and the last two require an appropriate zero padding.

Following~\Cref{l1:linear_chance} and as the components under the stochastic basis are affine over the $\mux{}, \muu{}, \varx{}, \varu{}$, each chance constraint translates to a second-order cone constraint.

\begin{equation}\label{eq:mixed_multi_lcc_lifted}
\scalebox{0.91}{%
$\begin{aligned}
    y_0^{(i)} = a_i^\top \mux{0:N} + \alpha_i^\top \muu{0:N-1}, \ 
    y^{(i)} = a_i^\top \mvarx{0:N} + \alpha_i^\top \mvaru{0:N-1} \\
    y_0^{(i)} + \Phi_{\mathcal N}^{-1} (1-\epsilon) \|y^{(i)} \|_I \le b 
\end{aligned}$%
}
\end{equation}

\subsection{Terminal Covariance Constraint}
The covariance constraint \eqref{eq:cov_control} in lifted representation is:
\begin{equation}
    \Vr{\state N }=\mvarx N \mvarx N ^\top \preceq \Sigma_f.
\end{equation}
Using $I_{\bas N} \succ 0$ and Schur's complement lemma, we obtain a semidefinite constraint on the decision variables
\begin{equation}
    \mvarx N \mvarx N ^\top \preceq \Sigma_f \iff \begin{bmatrix}
        \Sigma_f & \mvarx N \\ \mvarx N ^\top & I_{\bas N}
    \end{bmatrix} \succeq 0 \label{eq:cov_control_lifted}.
\end{equation}
\subsection{Conservative Approximation of Quadratic Chance Constraint to Covariance Constraint}

Let $\mathcal B(Q)$ denote the ellipsoidal set defined by $Q \in \pd{n_x}$,
\begin{equation}
    \mathcal B (Q) :=  \{x\in\R^{n_x}| x^\top Q^{-1}x \le 1 \}.
\end{equation}
Then, using set inclusion, it follows for $Q_1, Q_2 \in \pd {n_x}$:
\begin{equation}
    \mathcal B (Q_1) \subseteq \mathcal B (Q_2) \iff Q_1 \preceq Q_2 \label{eq:ellipse-set-inclusion}.
\end{equation}
Given a Gaussian R.V. $\state{} \sim \mathcal N (0, \Sigma)$, the $(1-\epsilon)$-confidence ellipsoidal set is $\mathcal B(\Phi_{\chi_{n_x}^2}^{-1}(1-\epsilon)\Sigma)$ as defined using $\chi_{n}^2$ distribution~\cite{lew_chance-constrained_2020}. 
Additionally, if the $(1-\epsilon)$-confidence ellipsoid is contained inside the ellipsoidal set formed by $Q$, then the probability constraint holds. Using \eqref{eq:ellipse-set-inclusion}, the following serves as a conservative approximation:

\begin{equation}
    \Phi_{\chi_n^2}^{-1}(1-\epsilon)\Sigma \preceq Q \implies 1-\epsilon \le \Pr{\state{}\in\mathcal{B}(Q)} .
\end{equation}
Hence, using the above and \eqref{eq:cov_control_lifted} for the terminal state, the following serves as a LMI approximation of the quadratic chance constraint \eqref{eq:termQ}:
\begin{equation}
    \begin{bmatrix}
        \tfrac 1 {\zeta_{Q}} \termQ & \mvarx N \\ \mvarx N ^\top & I_{\bas N}
    \end{bmatrix} \succeq 0 \label{eq:qcc_sdp_lifted},
\end{equation}
here, $\zeta_{Q} = \Phi_{\chi_n^2}^{-1}(1-\epsilon_N)$.

\begin{lemma}
    If $\eta \sim \N{0}{\Sigma}$ and
    \begin{equation}
        \Tr({\Sigma}) \le \frac{1}{\left(\Phi_{\mathcal N}^{-1}\left(\frac{1+\sqrt[n]{1-\epsilon}}{2}\right)\right)^2} \label{eq:l3:trace}
    \end{equation}
    then
    $
        \Pr{\|\eta\|_I \le 1} \ge 1-\epsilon
    $.\label{l3:QCC_to_Trace}
\end{lemma}
\begin{proof}
    Let the singular value decomposition of $\Sigma$ be denoted as $\Sigma = U \Lambda U^\top$, with $\{\lambda_i\}_{i=1}^n$ being the singular values, $\lambda_\star := \Tr(\Lambda)$ and the transformed vector $\xi := U^\top \eta$.
    
    Then, $\xi \sim \N{0}{\Lambda}$.
    Let $\mathcal C$ denote a hypercube set defined as  $$\mathcal{C} = \left\{x\in \R^n : |x_i| \le \sqrt\frac{{\lambda_i}}{{\lambda_\star}}\right\}. $$
    Using the separability of the integrals, we obtain:
    \begin{align*}
        \Pr{ \xi \in \mathcal C } &= \prod_{i=1}^n \left(
        \int_{-\sqrt{ \lambda_i/{\lambda_\star}}}^{\sqrt{\lambda_i/{\lambda_\star}}} \frac{\exp(-x^2/2\lambda_i)}{\sqrt{2\pi \lambda_i}} dx
        \right) \\
        & = \left(2\Phi_{\mathcal N}\left(\frac{1}{\sqrt{\lambda_\star}}\right) - 1\right)^n =: 1-f_{\mathcal N} (\lambda_\star)
    \end{align*}
    Therefore, $\lambda_\star \le f_{\mathcal N}^{-1}(\epsilon)
    $ if and only if $\Pr{\xi \in \mathcal C } \le 1-\epsilon$, as the function $f_{\mathcal N}(\lambda_\star)$ is monotonically increasing for $\lambda_\star>0$
    
    Let $\mathcal S = \{x\in \R^{n}: \|x\|_I \le 1\}$ denote the unit-sphere. It follows that the hypercube is contained within the sphere, i.e., $\mathcal C \subset \mathcal S$, therefore, $\Pr{ \xi \in \mathcal C } < \Pr{ \xi \in \mathcal S }$.

    As $\Pr{\eta \in \mathcal S} = \Pr{\xi \in \mathcal S}$ and $\Tr({\Sigma}) = \Tr(\Lambda)$, the claim follows.
\end{proof}
\subsection{Conservative Approximation of Quadratic Chance Constraint to Quadratic Constraints}
Given the constraint \eqref{eq:termQ} and the representation \eqref{eq:state_rep}, the constraint is
$
    \Pr{\| L_{\text{CC}} ^{-1}\mvarx{i} \xiv i\|_I^2 \le 1 } \ge 1 - \epsilon,
$
here $L_{\text{CC}} L_{\text{CC}}^\top = \termQ$. As $\eta = L_{\text{CC}} ^{-1}\mvarx{i} \xiv i$ is a Gaussian vector of size $n_x$, using~\Cref{l3:QCC_to_Trace}, we obtain a conservative quadratic constraint approximation on the decision variables,
\begin{equation}
        \varx i^\top (I_{\bas i} \otimes \termQ^{-1}) \varx i \le \frac{1}{\left(\Phi_{\mathcal N}^{-1}\left(\frac{1+\sqrt[n_x]{1-\epsilon}}{2}\right)\right)^2} \label{eq:qcc_trace_lifted},
\end{equation}
by using the trace trick 
and applying the standard vectorization identity $\Tr(A^\top B C) = \vc{A}^\top (C^\top \otimes B) \vc{C}$.
\begin{rmrk}
    The above constraint can also be set up similarly for a mixed multi-step quadratic chance constraint involving a linear combination of states and control inputs across multiple time steps.
\end{rmrk}
\begin{rmrk}
    Using Markov's concentration inequality~\cite{nakka_trajectory_2023}, one can obtain $\Tr(\Sigma) \le \epsilon$ instead of \eqref{eq:l3:trace} which is conservative as $\lim_{\epsilon \to 0^{+}} \epsilon^{-k}f_{\mathcal N}^{-1}(\epsilon) \to \infty, \forall k >0$. 
\end{rmrk}
\section{Main Theorem: Deterministic Convex Problem}
Using the above results we formulate a convex problem over the decision variables ${\mux{}, \muu{}, \varx{}, \varu{}}$ as follows:
\begin{problem}[\textbf{Lifted Convex Reformulation}]\label{prob:LDLOC}
\begin{subequations}

\begin{equation}
\scalebox{0.95}{
$\begin{split}
\min_{\mux{}, \muu{}, \varx{}, \varu{}} J = \sum_{k=0}^{N-1} \quadform{Q_k}{\mux k} + \quadform{(I_{\bas k} \otimes Q_k)}{\varx k} 
+ \quadform{R_k}{\muu k} \\
+ \quadform{(I_{\bas k} \otimes R_k)}{\varu k} 
+  \quadform{Q_N}{(\mux N-x^\star)} \\
+ \quadform{(I_{\bas N} \otimes Q_N)}{\varx N}
+ \sum_{i \in T_{wp}^\star} \quadform{Q_\star}{(\mux i - \mux i^\star)}
\end{split}
$}\label{eq:thm1:cost}
\end{equation}
\begin{align}
   \text{s.t.} \quad \eqref{eq:lifted_dynamics},\ \forall \ k\in \{0, 1, \dots, N-1\} \label{eq:thm1:sys-dyn}\\
    \eqref{eq:mixed_multi_lcc_lifted}\ \forall\ i \in  \itt{m}, \label{eq:thm1:lcc-lifted}\\
    \eqref{eq:cov_control_lifted}, \eqref{eq:qcc_sdp_lifted} \text{ or }\eqref{eq:qcc_trace_lifted} \label{eq:thm1:qcc_cov}.
\end{align} \label{eq:prob:lifted}
\end{subequations}
\end{problem}

\begin{theorem}[\textbf{Exact Convex Reformulation}]
\label{thm:exact_convex_reformulation}
Under~\Cref{asmp:independence_represent} and the lifted 
representation \eqref{eq:state_rep}, \eqref{eq:control_rep}, 
\Cref{prob:LDLOC} restricted to dynamics \eqref{eq:lifted_dynamics}, 
linear chance constraints \eqref{eq:mixed_multi_lcc_lifted}, and 
covariance constraint \eqref{eq:cov_control_lifted} is an exact 
deterministic convex reformulation of \Cref{prob:LSOC}. 
Specifically:
\begin{enumerate}[label=(\roman{*})]
    \item \textbf{Cost Equivalence:}
    \begin{equation}
        J_{\text{\Cref{prob:LSOC}}}(\state{}, \cntrl{}) = 
        J_{\text{\Cref{prob:LDLOC}}}(\mux{}, \muu{}, \varx{}, \varu{})
    \end{equation}
    
    \item \textbf{Feasible Set Equivalence:}
    \begin{equation}
    \scalebox{0.98}{%
    $\begin{split}
        (\mux{}, \muu{}, \varx{}, \varu{}) 
        \text{ feasible for \Cref{prob:LDLOC}} \iff \\ 
        \exists\ \text{affine policy } \pi \in \Pi_{\text{affine}} 
        \text{ feasible for \Cref{prob:LSOC}}
    \end{split}$%
    }
    \end{equation}
\end{enumerate}
\end{theorem}

\begin{proof}
(i) \textbf{Cost Equivalence:}
For the quadratic stage cost, expanding using \eqref{eq:state_rep}:
\begin{equation}
\begin{split}
    \Ex{\quadform{Q_k}{\state k}} &= 
    \Ex{(\mux k + \mvarx k \xiv k)^\top Q_k (\mux k + \mvarx k \xiv k)}\\
    &= \quadform{Q_k}{\mux k} 
 + \Ex{(\mvarx k \xiv k)^\top Q_k \mvarx k \xiv k}.
\end{split}
\end{equation}
The cross term vanishes as $\Ex{\xiv k} = 0$. Applying the 
trace trick and vectorization identity to the remaining term:
\begin{equation}
\begin{split}
    \Ex{(\mvarx k \xiv k)^\top Q_k \mvarx k \xiv k} 
    &= \varx{k}^\top (I_{\bas k} \otimes Q_k) \varx{k}.
\end{split}
\end{equation}
Hence:
$\Ex{\quadform{Q_k}{\state k}} = 
    \quadform{Q_k}{\mux k} + 
    \quadform{(I_{\bas k} \otimes Q_k)}{\varx k}$. Similarly for the control cost: $\Ex{\quadform{R_k}{\cntrl k}} = 
    \quadform{R_k}{\muu k} + 
    \quadform{(I_{\bas k} \otimes R_k)}{\varu k}$. Applying the same argument to the terminal cost and 
waypoint terms, \eqref{eq:thm1:cost} is obtained, 
establishing cost equivalence.

(ii) \textbf{Feasible Set Equivalence:} \textit{Forward direction} $(\Rightarrow)$: 
Given a feasible solution $(\mux{}, \muu{}, \varx{}, \varu{})$ 
to \Cref{prob:LDLOC}, define the affine policy 
$\pi \in \Pi_{\text{affine}}$ as:
\begin{equation}
    \cntrl k = \muu k + \mvaru k \xiv k
\end{equation}
By construction, this policy satisfies the dynamics 
\eqref{eq:sys_dyn} and~\Cref{asmp:independence_represent}. 
The linear chance constraints \eqref{eq:def_lcc_mixed_multi} 
are satisfied by \Cref{l1:linear_chance}, as the SOCP 
constraint \eqref{eq:mixed_multi_lcc_lifted} is an exact 
reformulation under Gaussian distributions. The covariance 
constraint \eqref{eq:cov_control} is satisfied by 
\eqref{eq:cov_control_lifted} via Schur's complement.

\textit{Backward direction} $(\Leftarrow)$: 
Given an affine policy $\pi \in \Pi_{\text{affine}}$ 
feasible for \Cref{prob:LSOC}, the lifted variables 
$(\mux{}, \muu{}, \varx{}, \varu{})$ are uniquely defined 
via \eqref{eq:state_rep}, \eqref{eq:control_rep} and 
satisfy \eqref{eq:lifted_dynamics} by construction. 
The satisfaction of \eqref{eq:mixed_multi_lcc_lifted} 
follows from \Cref{l1:linear_chance}, and 
\eqref{eq:cov_control_lifted} follows from 
\eqref{eq:state_mean_cov} and Schur's complement. 
Hence, the lifted variables are feasible for \Cref{prob:LDLOC}.
\end{proof}
\begin{rmrk}
Note that the optimality guarantee of \Cref{thm:exact_convex_reformulation} 
is restricted to the class of affine causal policies 
$\Pi_{\text{affine}}$ defined in~\Cref{asmp:independence_represent}.
\end{rmrk}
\begin{prop}[\textbf{Conservative Approximation of 
Quadratic Chance Constraints}]
\label{prop:conservative_qcc}
Under \Cref{asmp:independence_represent} and the lifted representation \eqref{eq:state_rep},\eqref{eq:control_rep}, both \eqref{eq:qcc_sdp_lifted} and \eqref{eq:qcc_trace_lifted} are conservative (sufficient) approximations of \eqref{eq:termQ}. Neither dominates the other: \eqref{eq:qcc_sdp_lifted} is tighter when the singular values of $\mvarx i \mvarx i^\top$ are nearly equal, while \eqref{eq:qcc_trace_lifted} is tighter when they are highly spread.
\end{prop}

\begin{proof}
For \eqref{eq:qcc_sdp_lifted}, the Gaussian R.V. 
$(\state i - \Ex{\state i}) \sim \mathcal{N}(0, 
\mvarx i \mvarx i^\top)$ has a $(1-\epsilon)$-confidence 
ellipsoidal set $\mathcal{B}(\zeta_Q \mvarx i \mvarx i^\top)$ 
where $\zeta_Q = \Phi_{\chi_{n_x}^2}^{-1}(1-\epsilon)$. 
Feasibility of \eqref{eq:qcc_sdp_lifted} implies 
$\mathcal{B}(\zeta_Q \mvarx i \mvarx i^\top) \subseteq 
\mathcal{B}(\termQ)$ via set inclusion 
\eqref{eq:ellipse-set-inclusion} and Schur's complement, 
hence \eqref{eq:termQ} is satisfied. For 
\eqref{eq:qcc_trace_lifted}, feasibility implies 
$\Tr(L_{\text{CC}}^{-1} \mvarx i \mvarx i^\top 
L_{\text{CC}}^{-\top}) \le f_{\mathcal{N}}^{-1}(\epsilon)$,  which by~\Cref{l3:QCC_to_Trace} implies \eqref{eq:termQ}. In both cases, the converse fails as the set inclusions are strict. 

For the relative tightness, \eqref{eq:qcc_sdp_lifted} is governed by $\lambda_{\max} \le \tfrac{1}{\zeta_Q}$ while \eqref{eq:qcc_trace_lifted} is governed by $\sum_j \lambda_j \le f_{\mathcal{N}}^{-1}(\epsilon)$, hence the dominance depends on the skewness of 
$\mvarx i \mvarx i^\top$ as stated. \end{proof}

\begin{figure}
\centering
    \begin{subfigure}[b]{0.48\textwidth}
    \centering
      \includegraphics[width=0.75\textwidth]{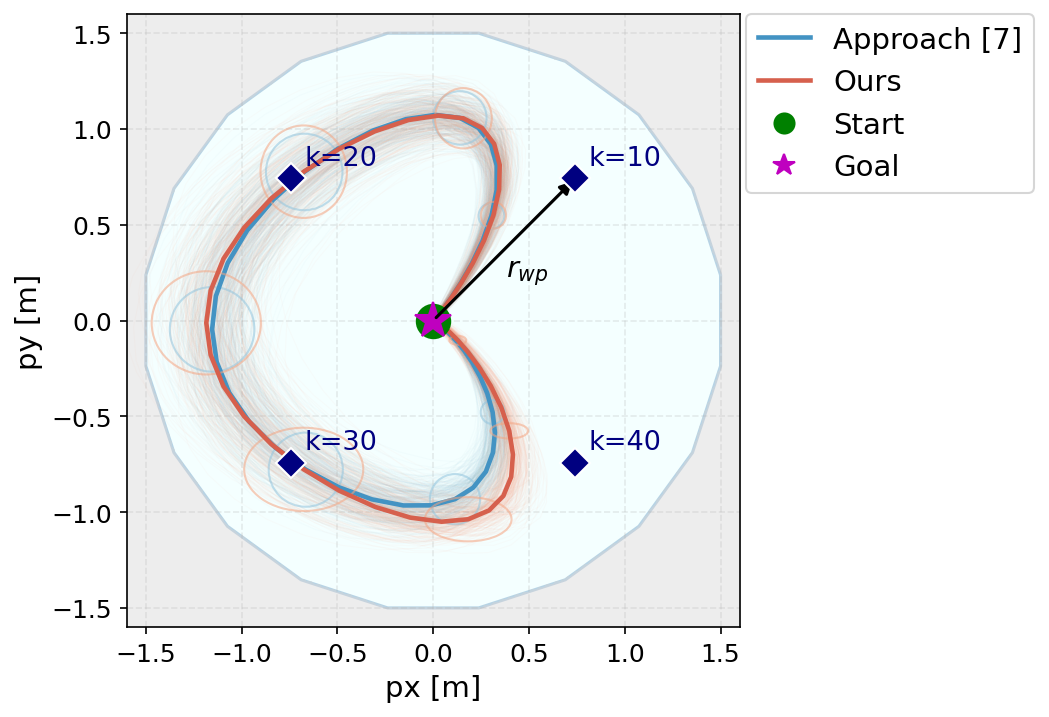}
      \caption{Circle Arena, ($r_{\text{wp}} = 1.05$) with $22\%$ cost reduction}
    \end{subfigure}
    \begin{subfigure}[b]{0.48\textwidth}
      \includegraphics[width=1.0\textwidth]{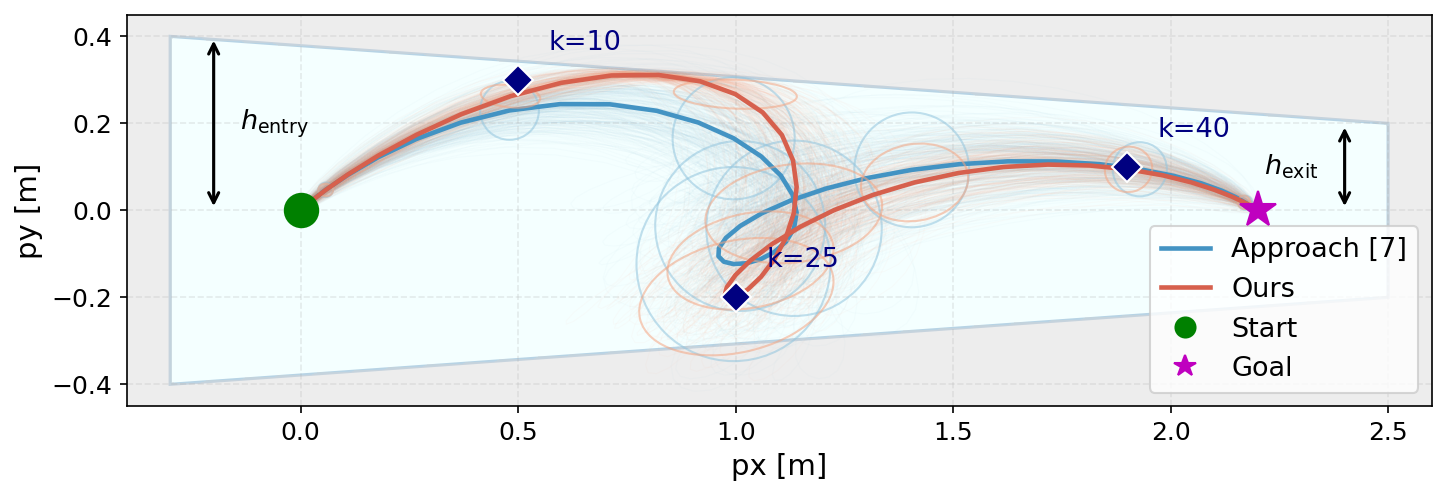}
      \caption{Funnel Corridor, $(h_{\text{entry}}, h_{\text{exit}})= (0.4, 0.2)$ with $32\%$ cost reduction}
    \end{subfigure}
    \caption{Minimum Snap Path Planning for the linearized lateral dynamics of a quadrotor under the two different setups under stochastic perturbation forces ($\sigma_1=\sigma_2 = 0.05$). The proposed approach leads to a cost reduction as compared to~\cite{lew_chance-constrained_2020}, as it can shape the covariance.}
    \label{fig:2:snap_traj}
\end{figure}
\begin{table}[ht] 
    \centering
    \caption{Performance comparison across varying noise levels ($\sigma$) and environment parameters for the Circle Arena and Funnel Corridor scenarios.  $\checkmark$ and $\times$ indicate successful completion and infeasibility, respectively.
    }
    \label{tab:parameter_sweep_summary}
    \begin{subtable}{\linewidth} 
    \centering
    \caption{Circle Arena}
    \label{tab:circle_arena}
    \scalebox{0.93}{%
    \begin{tabular}{c c c c c c}
        \toprule
        $\sigma$ & $r_{\text{wp}}$ & \textbf{\cite{lew_chance-constrained_2020}} & \textbf{\cite{rapakoulias_discrete-time_2023}} & \textbf{Ours} & $1 -\frac{J_{\text{ours}}}{J_{\text{\cite{lew_chance-constrained_2020}}}}$\\
        \midrule
        0.010 & 0.80 & $\checkmark$ & $\checkmark$ & $\checkmark$ & 4.43\% \\
        0.060 & 0.70 & $\checkmark$ & $\times$    & $\checkmark$ & 26.50\% \\
        0.060 & 0.50 & $\checkmark$ & $\times$    & $\checkmark$ & 35.48\% \\
        0.070 & 0.50 & $\times$    & $\times$    & $\checkmark$ & 100\% \\
        0.070 & 0.70 & $\times$    & $\times$    & $\checkmark$ & 100\% \\
        0.850 & 0.50 & $\times$    & $\times$    & $\checkmark$ & 100\% \\
        0.850 & 0.70 & $\times$    & $\times$    & $\checkmark$ & 100\% \\
        \bottomrule
    \end{tabular}%
    }
\end{subtable}

\vspace{1.0em}

\begin{subtable}{\linewidth}
    \centering
    \caption{Funnel Corridor}
    \label{tab:funnel_corridor}
    \scalebox{0.93}{%
    \begin{tabular}{c c c c c c c}
        \toprule
        $\sigma$ & $h_{\text{entry}}$ & $h_{\text{exit}}$ & \textbf{\cite{lew_chance-constrained_2020}} & \textbf{\cite{rapakoulias_discrete-time_2023}} & \textbf{Ours} & $1 - \frac{J_{\text{ours}}}{J_{\text{\cite{lew_chance-constrained_2020}}}}$ \\
        \midrule
        0.010 & 0.40 & 0.20 & $\checkmark$ & $\checkmark$ & $\checkmark$ & 4.47\% \\
        0.056 & 0.40 & 0.20 & $\checkmark$ & $\times$    & $\checkmark$ & 43.65\% \\
        0.056 & 0.50 & 0.30 & $\checkmark$ & $\times$    & $\checkmark$ & 24.21\% \\
        0.060 & 0.40 & 0.20 & $\times$    & $\times$    & $\checkmark$ & 100\% \\
        0.060 & 0.50 & 0.30 & $\times$    & $\times$    & $\checkmark$ & 100\% \\
        0.500 & 0.40 & 0.20 & $\times$    & $\times$    & $\checkmark$ & 100\% \\
        0.550 & 0.50 & 0.30 & $\times$    & $\times$    & $\checkmark$ & 100\% \\
        \bottomrule
    \end{tabular}%
    }
\end{subtable}

\vspace{1.0em}

\begin{subtable}{\linewidth}
    \centering
    \caption{Terminal Quadratic Chance Constraint}
    \label{tab3:term_qcc}
    \scalebox{0.93}{%
    \begin{tabular}{ll ccc}
\toprule
 $\sigma_1$ & $\sigma_2$
  & \multicolumn{1}{c}{Markov~\cite{nakka_trajectory_2023}}
  & \multicolumn{1}{c}{LMI \eqref{eq:qcc_sdp_lifted}}
  & \multicolumn{1}{c}{Quadratic \eqref{eq:qcc_trace_lifted}} \\
\midrule
$0.18$ & $0.005$ & $\times$            & $\times$                    & $\checkmark$  \\
$0.20$ & $0.010$ & $\times$            & $\times$                    & $\checkmark$ \\
$0.173$ & $0.173$ & $\times$             & $\checkmark$ & $\times$              \\
$0.175$ & $0.175$ & $\times$             & $\checkmark$ & $\times$              \\
\bottomrule
\end{tabular}%
}
\end{subtable}
\end{table}
\section{Numerical Simulations}
Consider the problem of generating a minimum-snap trajectory for a quadrotor in a 2D plane, where the lateral and longitudinal dynamics are modeled as two quadruple integrators, with random perturbation forces acting on them. The dynamics are obtained via a zero-order hold discretization:
\begin{equation}
\scalebox{0.80}{$
{A} = \begin{bmatrix}
I_2 & T I_2 & \frac{T^2}{2} I_2 & \frac{T^3}{6} I_2 \\
0_2 & I_2 & T I_2 & \frac{T^2}{2} I_2 \\
0_2 & 0_2 & I_2 & T I_2 \\
0_2 & 0_2 & 0_2 & I_2
\end{bmatrix},
{B} = \begin{bmatrix}
\frac{T^4}{24} I_2 \\ 
\frac{T^3}{6} I_2 \\ 
\frac{T^2}{2} I_2 \\ 
T I_2
\end{bmatrix},
{D} = \begin{bmatrix}
\frac{T^2}{2} I_2 \\ 
T I_2 \\ 
I_2 \\ 
0_2
\end{bmatrix}
\begin{bmatrix}
    \sigma_1 & 0 \\ 0 & \sigma_2
\end{bmatrix}
$} 
\end{equation}

here, $T=0.1s, \sigma_1, \sigma_2, 0_2, N =50$ are the discretization time, the noise parameter, the $2\times2$ zero matrix and the prediction horizon.
As shown in~\Cref{tab:lcc_comparison},~\cite{lew_chance-constrained_2020} and~\cite{rapakoulias_discrete-time_2023} account for chance constraints on state and control, we present comparisons with them.
\cite{lew_chance-constrained_2020} requires feedback gains, which we obtain by solving the Riccati equation. Similarly,~\cite{rapakoulias_discrete-time_2023} requires a covariance estimate, which we obtain via an iterative approach that improves the approximation using the previous solution, starting with an open-loop estimate.
The problem horizon is $N = 50$ steps at $T = \qty{0.1}{\second}$ ($\qty{5}{\second}$ total).

here, $u_{\max{}} = 25$, $\epsilon_u = 0.10$.
We consider two cases with varying parameters to demonstrate better optimality than~\cite{lew_chance-constrained_2020} and better feasibility than~\cite{lew_chance-constrained_2020,rapakoulias_discrete-time_2023}.

The initial state distribution is $\mathcal{N}(\mathbf{0}, \Sigma_0)$ with $\Sigma_0 = 10^{-5}I_8$. The objective is a quadratic cost with stage weight $Q$, control weight $R = I_2$, and terminal weight $Q_T$.

Linear chance constraints are imposed on the position, ensuring it remains within the boundary with $\epsilon=0.05$, and on control input such that $|\cntrl k^{(i)}|<u_m=25$ with $\epsilon_u = 0.05$ for each component and timestep. Soft constraints of waypoints is imposed using $Q_\star = \text{diag}(10^5,0,0,0)\otimes I_2$ for each case. 

\subsection{Circle Arena}
The feasible region is the interior of a circle of radius $R = \qty{1.5}{\meter}$, approximated as a convex polygon with $M = 20$ sides. The vehicle starts and ends at rest at the origin. Four soft waypoints are placed at radius $r_{\text{wp}}$ from the center. The stage cost penalizes deviations in all state components with $Q_{\text{circle}} = \text{diag}(1,0.1, 0.01,10^{-3}) \otimes I_2$ and a large terminal cost $Q_T = 10^6 \cdot \text{diag}(50, 2, 0.5, 0.1)\otimes I_2$. The sweep varies $\sigma$ and $r_{\text{wp}}$ to probe feasibility and cost suboptimality across noise levels and geometric difficulty.
\subsection{Funnel Corridor}
The feasible region is a trapezoidal corridor spanning $p_x \in [-0.3, 2.5]\, \text{m}$, defined by entry $h_{\text{entry}}$ and exit half-widths $h_{\text{exit}}$ using four position linear chance constraints. The vehicle starts at rest at the origin and reaches the goal state $x^\star = [2.2, 0_{1\times7}]^\top$. Three soft waypoints are set inside the narrowing passage. The stage cost suppresses velocity and higher derivatives only, with $Q_{\text{funnel}} = 0.05\text{diag}(0, 1, 0.1, 0.01) \otimes I_2$, and the terminal cost is $Q_T = 10^6 \text{diag}(80, 3, 0.5, 0.1)\otimes I_2$. The sweep varies $\sigma$, $h_{\text{entry}}$, and $h_{\text{exit}}$ to characterize solver behavior as the corridor narrows and noise increases.

Example scenarios are plotted in~\Cref{fig:2:snap_traj} and the results of the parameter sweep are reported in~\Cref{tab:parameter_sweep_summary}. We observe that the proposed approach remains feasible for much larger noise magnitudes, as it does not require an initial guess of covariances~\cite{rapakoulias_discrete-time_2023} or feedback matrices~\cite{lew_chance-constrained_2020}, and performs better than a decoupled approach~\cite{lew_chance-constrained_2020}.
\subsection{Terminal Quadratic Chance Constraint}
Under no linear chance constraints on position, for $\termQ = I_8$ and $\epsilon = 0.05$ in \eqref{eq:termQ}, we vary $\sigma_1, \sigma_2$ to demonstrate the feasibility and optimality in the LMI \eqref{eq:qcc_sdp_lifted}, QC \eqref{eq:qcc_trace_lifted}, and Markov's bound~\cite{nakka_trajectory_2023}. The stage costs, waypoints, terminal cost, and initial and terminal states are identical to the funnel corridor case. 
The findings in \Cref{tab3:term_qcc} demonstrate that Markov's leads to infeasibility easily and high costs, and additionally that the quadratic constraint is better in cases of high anisotropy in noise $\sigma_1\gg \sigma_2$; however, in cases of $\sigma_1 \approx \sigma_2$, the LMI constraint has the largest feasibility range.
\section{Conclusion}
In this paper, we presented a novel representation method that lifts the state vector to explicitly encode moment information, enabling exact representation of linear chance constraints as second-order cone constraints and less conservative approximate representations of quadratic chance constraints as linear-matrix and quadratic constraints over the decision variables. This representation leads to a deterministic convex formulation of the original stochastic optimal control problem.

We validated the framework on minimum-snap trajectory generation for a quadrotor. In the Circle Arena, the proposed approach achieves cost reductions of up to $35\%$ over~\cite{lew_chance-constrained_2020} and remains feasible at $\sigma= 0.85$, while both~\cite{lew_chance-constrained_2020} and~\cite{rapakoulias_discrete-time_2023} become infeasible beyond $\sigma \approx 0.06$. In the Funnel Corridor, we observe a cost improvement up to $43\%$ under moderate noise and feasibility at $\sigma=0.6$, while both baselines fail beyond $\sigma = 0.06$. 
Across all tested configurations, the proposed approach returned a feasible solution in every instance. We observe feasibility even with a noise standard deviation an order of magnitude higher than that of existing approaches, as the exact representation yields the true feasible region, thereby improving feasibility and cost optimality.
\bibliographystyle{IEEEtran}
\bibliography{main}{}
\end{document}